\documentclass[11pt]{article}

\usepackage{fullpage, amsmath, amsthm, amssymb, microtype, tikz}

\newtheorem{prop}{Proposition}
\newtheorem{thm}{Theorem}
\newtheorem{lem}{Lemma}

\theoremstyle{definition}
\newtheorem{defn}{Definition}

\newtheorem{remark}{Remark}

\newcommand{\N}{\mathbb{N}}

\renewcommand{\epsilon}{\varepsilon}

\renewcommand{\tilde}{\widetilde}

\renewcommand{\hat}{\widehat}

\DeclareMathOperator{\poly}{poly}
\mathchardef\mhyphen="2D

\newcommand{\defeq}{\stackrel{\text{def}}{=}}

\author{William M. Hoza\thanks{Department of Computer Science, University of Texas at Austin, \texttt{whoza@utexas.edu}. This research was mostly conducted while the author was an undergraduate student at the California Institute of Technology.} \and Chris Umans\thanks{Department of Computing and Mathematical Sciences, California Institute of Technology, \texttt{umans@cms.caltech.edu}}}
\title{Targeted Pseudorandom Generators, Simulation Advice Generators, and Derandomizing Logspace}

\begin{document}
	\maketitle

	\begin{abstract}
		Assume that for every derandomization result for logspace algorithms, there is a pseudorandom generator strong enough to nearly recover the derandomization by iterating over all seeds and taking a majority vote. We prove under a precise version of this assumption that $\mathbf{BPL} \subseteq \bigcap_{\alpha > 0} \mathbf{DSPACE}(\log^{1 + \alpha} n)$.
		
		We strengthen the theorem to an equivalence by considering two generalizations of the concept of a pseudorandom generator against logspace. A \emph{targeted pseudorandom generator} against logspace takes as input a short uniform random seed \emph{and} a finite automaton; it outputs a long bitstring that looks random to that particular automaton. A \emph{simulation advice generator} for logspace stretches a small uniform random seed into a long advice string; the requirement is that there is some logspace algorithm that, given a finite automaton and this advice string, simulates the automaton reading a long uniform random input. We prove that \[\bigcap_{\alpha > 0} \mathbf{promise\mhyphen BPSPACE}(\log^{1 + \alpha} n) = \bigcap_{\alpha > 0} \mathbf{promise\mhyphen DSPACE}(\log^{1 + \alpha} n)\] if and only if for every targeted pseudorandom generator against logspace, there is a simulation advice generator for logspace with similar parameters.
		
		Finally, we observe that in a certain \emph{uniform} setting (namely, if we only worry about sequences of automata that can be generated in logspace), targeted pseudorandom generators against logspace \emph{can} be transformed into simulation advice generators with similar parameters.
	\end{abstract}
	
	\section{Introduction}
	
	\subsection{Derandomization vs. pseudorandom generators}
	
	The \emph{derandomization} program of complexity theory consists of trying to deterministically simulate whole classes of randomized algorithms without significant loss in efficiency. For example, we would like to prove that $\mathbf{P} = \mathbf{BPP}$, $\mathbf{NP} = \mathbf{AM}$, and $\mathbf{L} = \mathbf{BPL}$. The main strategy for derandomization is to design an efficient \emph{pseudorandom generator}. A natural question is whether this strategy is without loss of generality. That is, does derandomization always imply a pseudorandom generator that is strong enough to recover that very same derandomization? This question appears to have first been investigated by Fortnow \cite{for01}, who gave an oracle separation between pseudorandom generators and derandomization in the $\mathbf{P}$ vs. $\mathbf{BPP}$ setting.
	
	Nevertheless, for both $\mathbf{NP}$ vs. $\mathbf{AM}$ and $\mathbf{P}$ vs. $\mathbf{BPP}$, there are indeed known constructions of pseudorandom generators from derandomization assumptions. Most such constructions come from the \emph{hardness vs. randomness} paradigm. The idea is to show that derandomization assumptions imply \emph{hardness} results (such as circuit lower bounds). There is a large body of literature \cite{yao82, bm84, nw94, iw97, iw98, hill99, kvm02, uma03} showing how, in turn, to construct pseudorandom generators from hardness. Typically, the constructed pseudorandom generator is not strong enough to recover the original derandomization assumption (e.g. \cite{ikw02, ki04, aghk11, kvms12, wil13}) but some results are known that establish exact equivalence between certain sorts of derandomizations and certain sorts of pseudorandom generators (see \cite{avm12}). Goldreich has followed another approach \cite{gol11, gol11b} to construct pseudorandom generators from derandomization assumptions in the $\mathbf{BPP}$ setting. His approach does not directly involve establishing hardness results on the way; instead, he shows how to derandomize the standard nonconstructive existence proof for pseudorandom generators by a reduction to decision problems.
	
	The subject of this paper is $\mathbf{L}$ vs. $\mathbf{BPL}$. In this setting, there are no known constructions of pseudorandom generators from generic derandomization assumptions. Further, the question of whether derandomization is equivalent to pseudorandom generators is especially well-motivated in this setting, because nontrivial derandomizations and pseudorandom generators have been unconditionally constructed -- and there is a significant \emph{gap}. Iterating over all seeds of the best known pseudorandom generator, by Nisan \cite{nis92}, merely proves that $\mathbf{BPL} \subseteq \mathbf{DSPACE}(\log^2 n)$ (which can also be proven by recursive matrix exponentiation). But the best known derandomization, the celebrated Saks-Zhou theorem \cite{sz99}, states that $\mathbf{BPL} \subseteq \mathbf{DSPACE}(\log^{3/2} n)$.
	
	In this work, we show that (informally) \emph{if} for every derandomization of logspace algorithms, there is a pseudorandom generator strong enough to nearly recover the derandomization by iterating over all seeds, then $\mathbf{BPL} \subseteq \bigcap_{\alpha > 0} \mathbf{DSPACE}(\log^{1 + \alpha} n)$. So establishing the \emph{equivalence} of derandomization and pseudorandom generators \emph{would itself} yield a strong derandomization of $\mathbf{BPL}$.
	
	Our result can be viewed pessimistically as showing that it will be challenging to establish equivalence of derandomization and pseudorandom generators in the $\mathbf{BPL}$ setting. But it can also be viewed optimistically as giving a \emph{road map} for proving that $\mathbf{BPL} \subseteq \bigcap_{\alpha > 0} \mathbf{DSPACE}(\log^{1 + \alpha} n)$. From this second viewpoint, our result should be compared to other known results that give interesting sufficient conditions for derandomizing logspace:
	\begin{itemize}
		\item Klivans and van Melkebeek showed \cite{kvm02} that if some language in $\mathbf{DSPACE}(n)$ requires branching programs of size $2^{\Omega(n)}$, then there is a pseudorandom generator strong enough to prove $\mathbf{L} = \mathbf{BPL}$. While interesting, this result does not seem to provide a viable road map for derandomizing logspace, because the strong hardness assumption seems to be far beyond current understanding.
		\item Reingold, Trevisan, and Vadhan showed \cite{rtv06} that if there is an efficient pseudorandom walk generator for \emph{regular} digraphs, then $\mathbf{L} = \mathbf{RL}$. This result \emph{can} be reasonably thought of as giving a road map for derandomizing logspace; the result is particularly tantalizing because in the same work, they actually \emph{did} construct a pseudorandom walk generator for \emph{consistently labeled} regular digraphs. Alas, in the decade since these results were announced, nobody has been able to close the gap.
	\end{itemize}
	We view our result as promising, considering that there are already established techniques for proving equivalence of derandomization and pseudorandom generators. We consider it conceivable that those techniques can be ``ported'' to the $\mathbf{L}$ vs. $\mathbf{BPL}$ setting. The previously mentioned result of \cite{kvm02} may be a first step in that direction. To put it another way, for decades, researchers have been trying to design strong pseudorandom generators for $\mathbf{BPL}$; our result shows that researchers can feel free to \emph{make derandomization assumptions} while trying to design those pseudorandom generators, which could make the task significantly easier.
	
	\subsection{Four types of derandomization}
	
	In fact, our main result is considerably stronger than what we have said so far. To explicate our main result, it is useful to distinguish between \emph{four} types of derandomizations of logspace. (See Figure~\ref{fig:derand}.) First, the most generic type of derandomization is a \emph{simulator} for logspace. This is an algorithm that takes as input a finite automaton $Q$, a start state $q$, and a short uniform seed $x$; it outputs a state $\mathsf{Sim}(Q, q, x)$ whose distribution is close to the distribution of final states that $Q$ would be in were it to read a long uniform random string. (Finite automata provide a simple nonuniform model of space-bounded computation; each state of a $w$-state automaton corresponds to a configuration of a $(\log w)$-space Turing machine.)
	\begin{figure}
		\centering
		\begin{tikzpicture}
			\node[text width=2.5cm, align=center] (a) at (0, 0) {Ordinary pseudorandom generator};
			
			\node[align=center] (b) at (0, 6) {Simulator};
			
			\node[text width=2.5cm, align=center] (c) at (-3, 3) {Targeted pseudorandom generator};
			
			\node[text width=2cm, align=center] (d) at (3, 3) {Simulation advice generator};
			
			\draw[->, >=latex] (a) to (c);
			\draw[->, >=latex] (a) to (d);
			\draw[->, >=latex] (c) to (b);
			\draw[->, >=latex] (d) to (b);
			\draw[->, >=latex, dashed] (c) to (d);
		\end{tikzpicture}
		
		\caption{The four types of derandomization that we consider. A solid arrow from $A$ to $B$ indicates that a derandomization of type $A$ trivially implies a derandomization of type $B$. Our main result is that the implication indicated by the dashed arrow is equivalent to the statement that $\bigcap_{\alpha > 0} \mathbf{promise\mhyphen BPSPACE}(\log^{1 + \alpha} n) = \bigcap_{\alpha > 0} \mathbf{promise\mhyphen DSPACE}(\log^{1 + \alpha} n)$.} \label{fig:derand}
	\end{figure}
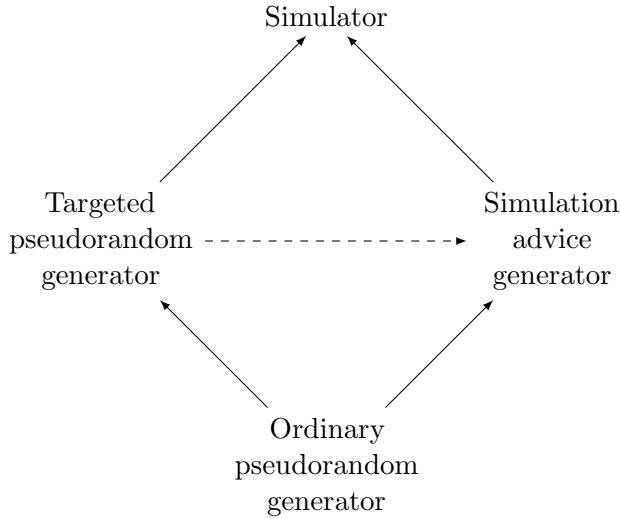
	
	The second type of derandomization, which should be familiar, is a \emph{pseudorandom generator} against logspace. A pseudorandom generator has two key features that distinguish it from a generic simulator:
	\begin{itemize}
		\item Input. The pseudorandom generator does not get to see the ``source code'' of the algorithm being simulated, i.e. it does not get $(Q, q)$ as part of its input.
		\item Output. The pseudorandom generator produces a long string for the automaton to read, whereas a simulator merely produces the final state of the automaton.
	\end{itemize}
	The third and fourth types of derandomization that we will consider generalize the concept of a pseudorandom generator by relaxing these two features respectively. The third type of derandomization, a \emph{targeted pseudorandom generator}, gets as input a finite automaton $Q$, a start state $q$, and a short uniform seed $x$; it outputs a long bitstring $\mathsf{Gen}(Q, q, x)$ that looks random to that particular automaton $Q$ when it starts in that particular state $q$. (Goldreich \cite{gol11b} coined the term ``targeted pseudorandom generator'' in the context of $\mathbf{P}$ vs. $\mathbf{BPP}$, where the generator gets a Boolean circuit as its auxiliary input. In the $\mathbf{L}$ vs. $\mathbf{BPL}$ setting, targeted pseudorandom generators have been studied before; see e.g. \cite{nis92b, rr99}.) The fourth type of derandomization, a \emph{simulation advice generator}, stretches a short uniform seed $x$ into a long advice string $\mathsf{Gen}(x)$; the requirement is that there is a deterministic logspace algorithm $\mathsf{S}$ such that $\mathsf{Sim}(Q, q, x) \defeq \mathsf{S}(Q, q, \mathsf{Gen}(x))$ is a simulator for logspace. To the best of our knowledge, we are the first to study simulation advice generators.
	
	Our main result is that
	\begin{equation}
		\bigcap_{\alpha > 0} \mathbf{promise\mhyphen BPSPACE}(\log^{1 + \alpha} n) = \bigcap_{\alpha > 0} \mathbf{promise\mhyphen DSPACE}(\log^{1 + \alpha} n) \label{eqn:derand}
	\end{equation}
	if and only if for every targeted pseudorandom generator against logspace, there is a simulation advice generator with similar parameters. (The precise statement is in Section~\ref{sec:main-result}.) Here, $\mathbf{promise\mhyphen BPSPACE}(s(n))$ is the set of promise problems decidable by probabilistic space-$s(n)$ Turing machines that always halt and that have error probability at most $1/3$; $\mathbf{promise\mhyphen DSPACE}(s(n))$ is its deterministic analog.
	
	Additionally, in Section~\ref{sec:uniform}, we observe that targeted pseudorandom generators against logspace \emph{can} be transformed into simulation advice generators for logspace if we move to the \emph{uniform setting}, i.e. we only worry about sequences of automata that can be generated in logspace. This is almost immediate from the definitions, but it illustrates how much easier it is to construct simulation advice generators than it is to construct pseudorandom generators.
	
	\subsection{Proof techniques} \label{sec:proof-overview}
	
    One direction of our main result is easy. Under the assumption that Equation~\ref{eqn:derand} holds, simulation advice generators are uninteresting objects that can be constructed for trivial reasons. The main content of the theorem is the reverse direction.
	
	The proof of the harder direction is by extending the techniques of Saks and Zhou \cite{sz99}. The way Saks and Zhou originally presented their result is that they used specific properties of Nisan's pseudorandom generator \cite{nis92} to design a space-efficient algorithm for approximate matrix exponentiation by reusing parts of the seed. Later, Armoni \cite{arm98} constructed a pseudorandom generator that is better than Nisan's for fooling low-randomness algorithms, and using Zuckerman's oblivious sampler \cite{zuc97}, he adapted the Saks-Zhou algorithm to use his generator instead of Nisan's, giving a better derandomization of such algorithms.
	
	In Section~\ref{sec:saks-zhou}, we show that with Armoni's ideas, the Saks-Zhou construction can instead be formulated as a \emph{transformation on simulators.} Roughly: Starting from a simulator that uses an $s$-bit seed to simulate $m_0$ steps of a $w$-state automaton, given a parameter $m$, the Saks-Zhou-Armoni (SZA) transformation produces a new simulator that uses an $O\left(s + \frac{(\log m) (\log w)}{\log m_0}\right)$-bit seed to simulate $m$ steps of a $w$-state automaton. We consider this reformulation to be interesting in its own right, as it clarifies the power of Saks-Zhou rounding.
	
	A simple, tempting idea is to start with a weak simulator and apply the SZA transformation $t$ times for some large constant $t$. In iteration $i$, choose $m = 2^{\log^{i/t} w}$. Then we end up with a simulator with $m = w$ (large enough to simulate randomized space-bounded algorithms), and the seed length is only $O(\log^{1 + 1/t} w)$! But unfortunately, the space complexity blows up with each application of the SZA transformation.
	
	Because of the recursive structure of the SZA transformation, the blowup can be avoided as long as the SZA transformation is only applied to simulators obtained from simulation advice generators. So to prove the harder direction of our main result, we cycle between three transformations:
	\begin{enumerate}
		\item Our assumption, which transforms a targeted pseudorandom generator into a simulation advice generator. (This ``transformation'' is not necessarily effective.)
		\item The SZA transformation, which we now think of as transforming a simulation advice generator into a simulator.
		\item A simple transformation based on the method of conditional probabilities, which transforms a simulator into a targeted pseudorandom generator.
	\end{enumerate}
	The SZA transformation substantially increases the number of steps being simulated. For each of the three transformations, we incur only mild degradation in the seed length, space complexity, etc. Hence, overall, each cycle significantly increases the output length of our targeted pseudorandom generator without degrading the other parameters too much. By iterating the cycle a large constant number of times, we end up with a generator strong enough to collapse $\bigcap_{\alpha > 0} \mathbf{promise\mhyphen BPSPACE}(\log^{1 + \alpha} n)$ to $\bigcap_{\alpha > 0} \mathbf{promise\mhyphen DSPACE}(\log^{1 + \alpha} n)$.
	
	\section{Formal statement of main result} \label{sec:main-result}

	Let $[w]$ denote the set $\{1, 2, \dots, w\}$. Let $U_n$ denote the uniform distribution on $\{0, 1\}^n$. For two probability distributions $\mu, \mu'$ on the same measurable space, write $\mu \sim_{\epsilon} \mu'$ to mean that the total variation distance between $\mu$ and $\mu'$ is at most $\epsilon$.

	\begin{defn}
		If $\mathbf{A}$ is a set of functions $\{0, 1\}^m \to [w]$, we say that a function $\mathsf{Sim}: \mathbf{A} \times \{0, 1\}^s \to [w]$ is an \emph{$\epsilon$-simulator} for $\mathbf{A}$ if for every $f \in \mathbf{A}$, we have $\mathsf{Sim}(f, U_s) \sim_{\epsilon} f(U_m)$.
	\end{defn}

	\begin{defn}
		If $\mathbf{A}$ is a set of functions $\{0, 1\}^m \to [w]$, we say that a function $\mathsf{Gen}: \mathbf{A} \times \{0, 1\}^s \to \{0, 1\}^m$ is a \emph{targeted $\epsilon$-pseudorandom generator} against $\mathbf{A}$ if the function $\mathsf{Sim}(f, x) \defeq f(\mathsf{Gen}(f, x))$ is an $\epsilon$-simulator for $\mathbf{A}$.
	\end{defn}
	The standard definition of a pseudorandom generator is the special case where $\mathsf{Gen}(f, x)$ does not depend on $f$.
	\begin{defn}
		A \emph{$(w, d)$-automaton} is a function $Q: [w] \times \{0, 1\}^d \to [w]$. If $Q_1$ is a $(w, d_1)$-automaton and $Q_2$ is a $(w, d_2)$-automaton, then $Q_2 Q_1$ is the $(w, d_1 + d_2)$-automaton defined by
		\[
		(Q_2 Q_1)(q; x, y) = Q_2(Q_1(q; x); y).
		\]
		Let $\mathbf{Q}_{w, d}^m$ be the set of all functions $\{0, 1\}^{md} \to [w]$ of the form $x \mapsto Q^m(q; x)$ where $Q$ is a $(w, d)$-automaton.
	\end{defn}
	
	In words, $\mathbf{Q}_{w, d}^m$ is the set of functions computed by letting a $(w, d)$-automaton run for $m$ steps and observing its final state. An element of $\mathbf{Q}_{w, d}^m$ can be specified by a pair $(Q, q)$, and this is how it will be presented to simulators and targeted pseudorandom generators in our theorem statements.
	
	\begin{defn}
		Suppose that for each $w$, $\mathsf{Gen}_w : \{0, 1\}^s \to \{0, 1\}^a$ is a function, and $\mathbf{A}_w \subseteq \mathbf{Q}_{w, d}^m$, where $s, a, d, m$ are functions of $w$. We say that $\mathsf{Gen}_w$ is\footnote{Strictly speaking, this is a property of the family $\{\mathsf{Gen}_w\}$, not of the individual function. There should be just one $\mathsf{S}$ for the whole family, and $\epsilon$ is a function of $w$.} an \emph{$\epsilon$-simulation advice generator for $\mathbf{A}_w$} if there is some deterministic logspace algorithm $\mathsf{S}$ such that the function $\mathsf{Sim}(Q, q, x) \defeq \mathsf{S}(Q, q, \mathsf{Gen}_w(x))$ is an $\epsilon$-simulator for $\mathbf{A}_w$.
	\end{defn}

	\begin{figure}
		\centering
		\bgroup
		\def\arraystretch{1.1}%
		\begin{tabular}{|r|l|}
			\hline
			Parameter & Interpretation \\
			\hline
			$w$ & Number of states in the automaton \\
			$d$ & Number of bits the automaton reads in each step \\
			$m$ & Number of steps the automaton takes \\
			$\epsilon$ & Simulation error, in total variation distance \\
			$s$ & Seed length \\
			$a$ & Number of advice bits \\
			\hline
		\end{tabular}
		\egroup
		\caption{A summary of the parameters of the targeted pseudorandom generators, simulation advice generators, and simulators that we study. Each family of generators/simulators is indexed by $w$, and the other parameters are functions of $w$.}
	\end{figure}
	
	Note that $\mathsf{S}$'s space bound is logarithmic in terms of its input length, i.e. it may use $O(d + \log w + \log a)$ bits of space. It is desirable for $m$ to be \emph{big} and $s, a, \epsilon$ to be \emph{small}. E.g. as long as $a \leq \poly(w, 2^d)$, it contributes nothing to the asymptotic space complexity of $\mathsf{S}$. To explicate the definition, we give several examples of where simulation advice generators might come from: 
	\begin{enumerate}
		\item Any (standard, non-targeted) $\epsilon$-pseudorandom generator $\mathsf{Gen}_w$ against $\mathbf{Q}_{w, d}^m$ is also an $\epsilon$-simulation advice generator for $\mathbf{Q}_{w, d}^m$. The associated algorithm $\mathsf{S}(Q, q, y)$ computes $Q^m(q; y)$ where $y$ is the output of $\mathsf{Gen}_w$. This can be done in logspace by storing the current state of $Q$ and the current $d$-bit chunk of $y$.
		\item Suppose there is some logspace $\epsilon$-simulator for $\mathbf{Q}_{w, d}^m$ with seed length $s$. Then the identity function on $\{0, 1\}^s$ is an $\epsilon$-simulation advice generator for $\mathbf{Q}_{w, d}^m$. (So under the assumption that $\mathbf{promise\mhyphen L} = \mathbf{promise\mhyphen BPL}$, simulation advice generators are only interesting for extreme values of parameters.)
		\item Suppose $\mathsf{Gen}_w$ is a targeted $\epsilon$-pseudorandom generator against $\mathbf{Q}_{w, d}^m$ of the form $\mathsf{Gen}_w(Q, q, x) = \mathsf{G}(\mathsf{Compress}(Q, q, x), x)$, where $\mathsf{Compress}$ is computable in $O(d + \log w)$ space and outputs $b$ bits. Let $\mathsf{Gen}'_w(x)$ be $x$ concatenated with the truth table $T$ of $\mathsf{G}(\cdot, x)$. Then $\mathsf{Gen}'_w$ is an $\epsilon$-simulation advice generator for $\mathbf{Q}_{w, d}^m$ with output length $a = s + m 2^b$. The associated algorithm $\mathsf{S}(Q, q, x, T)$ computes $c = \mathsf{Compress}(Q, q, x)$, referring to its advice tape for access to $x$. Then, $\mathsf{S}$ looks up the value $y = G(c, x)$ in the $T$ portion of its advice tape and computes $Q^m(q; y)$.
		\item Suppose $\mathsf{Sim}$ is an $\epsilon$-simulator for $\mathbf{Q}_{w, d}^m$ that perhaps uses much more than logspace, but that, each time it reads from $Q$ or $q$, first \emph{erases} all but $O(d + \log w)$ bits. If $c$ is a configuration of $\mathsf{Sim}(Q, q, x)$ in which $\mathsf{Sim}$ just read from $Q$ or $q$, then let $f(c, x)$ be the configuration that $\mathsf{Sim}(Q, q, x)$ will next be in when it is about to read from $Q$ or $q$. Let $\mathsf{Gen}_w(x)$ be the truth table of $f(\cdot, x)$. Then $\mathsf{Gen}_w$ is an $\epsilon$-simulation advice generator for $\mathbf{Q}_{w, d}^m$ with output length $a \leq \poly(w, 2^d)$. The associated algorithm $\mathsf{S}(Q, q, \mathsf{Gen}_w(x))$ simulates $\mathsf{Sim}(Q, q, x)$. To update the simulation's configuration, $\mathsf{S}$ alternates between reading a bit from $(Q, q)$ and using its advice tape.
	\end{enumerate}
    
    Suppose $\{\mathsf{F}_w\}$ is a family where $\mathsf{F}_w$ is a simulator for, a simulation advice generator for, or a targeted pseudorandom generator against $\mathbf{Q}_{w, d}^m$, with seed length $s(w)$. For convenience, we will say that the family is \emph{efficiently computable} if $s(w)$ is space constructible and given $(w, X)$, $\mathsf{F}_w(X)$ can be computed in deterministic space $O(s(w))$. We will often speak of an individual function $\mathsf{F}_w$ being efficiently computable when the family is clear.
	
	We now formally state our main result. In Condition~\ref{cond:targeted-prg-to-advice-generator}, $\eta, \sigma, \mu$ are the parameters of the targeted pseudorandom generator. The last parameter $\gamma$ quantifies the extent to which the derandomization degrades when the targeted pseudorandom generator is replaced with a simulation advice generator.
	
	\begin{samepage}
	\begin{thm} \label{thm:targeted-prg-to-advice-generator}
		The following are equivalent.
		\begin{enumerate}
			\item 	\[
			\bigcap_{\alpha > 0} \mathbf{promise\mhyphen BPSPACE}(\log^{1 + \alpha} n) = \bigcap_{\alpha > 0} \mathbf{promise\mhyphen DSPACE}(\log^{1 + \alpha} n).
			\] \label{cond:derand}
			\item For any constant $\mu \in [0, 1]$, for any sufficiently small constants $\sigma > \eta > 0$, and for any constant $\gamma > 0$, the following holds. Suppose there is a family $\{\mathsf{Gen}_w\}$, where $\mathsf{Gen}_w$ is an efficiently computable targeted $\epsilon$-pseudorandom generator against $\mathbf{Q}_{w, 1}^m$ with seed length $s$, satisfying
			\begin{align*}
				&s \leq O(\log^{1 + \sigma} w), 
				&&\log(1/\epsilon) = \log^{1 + \eta} w,
				&&&\log m \geq \log^{\mu} w.
				&&&& \\
				\intertext{Then there is another family $\{\mathsf{Gen}'_w\}$, where $\mathsf{Gen}'_w$ is an efficiently computable $\epsilon'$-simulation advice generator for $\mathbf{Q}_{w, 1}^{m'}$ with seed length $s'$ and output length $a'$, satisfying}
				&s' \leq O(\log^{1 + \sigma + \gamma} w),
				&&\log(1/\epsilon') = \log^{1 + \eta - \gamma} w,
				&&&\log m' \geq \log^{\mu - \gamma} w,
				&&&&\log a' \leq O(\log^{1 + \eta + \gamma} w).
			\end{align*}
			 \label{cond:targeted-prg-to-advice-generator}
		\end{enumerate}
	\end{thm}
	\end{samepage}
	
	\section{The implicit oracle model} \label{sec:implicit-oracle}

	Toward proving Theorem~\ref{thm:targeted-prg-to-advice-generator}, we introduce a model of space-bounded oracle algorithms that seemingly does not appear in the literature. Our new oracle model (the ``implicit oracle model'') gives a convenient framework for expressing the SZA result as a transformation on simulators and clarifies the effect on the simulator's space complexity when the SZA transformation is iterated.
	
	The implicit oracle model is similar to Wilson's oracle stack model \cite{wil88}, and it is appropriate for the situation where the algorithm doesn't have room to write down the entire query string, but it is ready to provide the oracle with random access to the query string (possibly by making more oracle queries.)
	\begin{defn}
		Fix a set $A \subseteq \{0, 1\}^*$. Giving an algorithm \emph{implicit oracle access} to $A$ allows the algorithm to interface with an oracle in the following ways:
		\begin{itemize}
			\item The algorithm can \emph{invoke} the oracle, which passes control to the oracle.
			\item The oracle can \emph{read position $j \in \N$} by giving $j$ to the algorithm. This passes control back to the algorithm. We associate this read with the most recent unresolved invocation.
			\item The algorithm can give the oracle a \emph{query value $b \in \{0, 1, \bot\}$}. This passes control back to the oracle and \emph{resolves} the most recent unresolved read.
			\item The oracle can give the algorithm a boolean \emph{answer value}. This passes control back to the algorithm and \emph{resolves} the most recent unresolved invocation.
		\end{itemize}
		The oracle is guaranteed to behave as follows: Fix any $x \in \{0, 1\}^*$. Suppose that for some invocation, when the oracle reads position $j$, the algorithm specifies value $x_j$ (where we interpret $x_j = \bot$ for $j > |x|$.) Then the oracle will make finitely many reads and give the answer value corresponding to whether $x \in A$, and every read will be of a position $j \leq |x| + 1$.
	\end{defn}
	We extend the definition by saying that we give an algorithm implicit oracle access to a \emph{function} $f: \{0, 1\}^* \to \{0, 1\}^*$ to mean that we give the algorithm implicit oracle access to the set $A = \{(x, b, 0) : |f(x)| \leq b\} \cup \{(x, b, 1) : f(x)_b = 1\}$.
	
	Wilson's oracle stack model is equivalent to the implicit oracle model with the additional restriction that the oracle is guaranteed to read its input from left to right.
	
	Ultimately, we will only use the implicit oracle model in intermediate steps of our proof; for our final algorithm, we will ``plug in'' actual algorithms in place of the oracle. The next lemma says what happens to space complexity when this actual algorithm is plugged in.
	
	\begin{lem} \label{lem:implicit-oracle}
		Suppose $\mathsf{Gen}: \{0, 1\}^s \to \{0, 1\}^a$ is an efficiently computable $\epsilon$-simulation advice generator for $\mathbf{Q}_{w, d}^m$, and let $\mathsf{Sim}$ be the corresponding simulator. Suppose $\mathsf{Alg}$ is an implicit oracle algorithm and $x$ is an input such that during the execution of $\mathsf{Alg}^{\mathsf{Sim}}(w, x)$, $\mathsf{Alg}$ uses $s'$ bits of space, and at any moment, there are at most $u$ unresolved oracle invocations, and there are at most $v$ unresolved reads of seeds. Then $\mathsf{Alg}^{\mathsf{Sim}}(w, x)$ can be computed (by a non-oracle algorithm) in space $O(s' + s \cdot (v + 1) + u \cdot (d + \log w + \log a))$.
	\end{lem}
	
	\begin{proof}
		Recall that $\mathsf{Sim}$ is of the form $\mathsf{Sim}(Q, q, x) = \mathsf{S}(Q, q, \mathsf{Gen}(x))$. Naturally, just simulate $\mathsf{Alg}$, replacing its oracle queries with computations of $\mathsf{Sim}$. The space needed is $s'$ for the computation of $\mathsf{Alg}$, plus $O(d + \log w + \log a)$ for each unresolved execution of $\mathsf{S}$, plus $O(s)$ for each unresolved execution of $\mathsf{Gen}$. The number of unresolved executions of $\mathsf{S}$ is precisely $u$. The number of unresolved executions of $\mathsf{Gen}$ is at most $v + 1$, because while an instance of $\mathsf{Sim}$ is in the process of computing $\mathsf{Gen}$, that instance never queries the $(Q, q)$ portion of its input.
	\end{proof}
	
	\section{The SZA transformation} \label{sec:saks-zhou}
	Formulating the Saks-Zhou construction as a transformation on simulators is not technically challenging. A \emph{$(w, d)$-automaton with fail state}\footnote{This is equivalent to the definition of a ``finite state machine of type $(w, d)$'' in \cite{sz99} or that of a ``$(w, d)$-automaton'' in \cite{ccvm06}.} is a $(w + 1, d)$-automaton such that $Q(w + 1; y) = w + 1$ for all $y$. (We think of $w + 1$ as the ``fail state''.) Let $\tilde{\mathbf{Q}}_{w, d}^m$ be the set of all functions of the form $x \mapsto Q^m(q; x)$ where $Q$ is a $(w, d)$-automaton with fail state. When we give an algorithm (implicit) oracle access to an $\epsilon$-simulator for $\tilde{\mathbf{Q}}_{w, d}^m$ with seed length $s$, it is understood that the algorithm can query for the parameters $w, d, m, \epsilon, s$ as well as interacting with the oracle in the usual way.
	
	\begin{samepage}
		\begin{thm} \label{thm:sza}
			There is a constant $c \in \N$ and a deterministic implicit oracle algorithm $\mathsf{SZA}$ with the following properties. Pick $w \in \N, \epsilon > 0$ and let $d = \lceil c \log(w/\epsilon) \rceil$. Suppose $\mathsf{Sim}$ is an $\epsilon$-simulator for $\tilde{\mathbf{Q}}_{w, d}^{m_0}$ with seed length $s \leq m_0 \leq w$. Then
			\begin{enumerate}
				\item For any $m \in \N$, there is some $m' \geq m$ such that $\mathsf{SZA}_m^{\mathsf{Sim}}$ is a $(12m\epsilon)$-simulator for $\tilde{\mathbf{Q}}_{w, d}^{m'}$. (Here $m$ is an input to $\mathsf{SZA}$; we write it as a subscript merely to separate it from the usual simulator inputs.)
				\item At any moment in the execution of $\mathsf{SZA}_m^{\mathsf{Sim}}$, there are at most $u \defeq \lceil (\log m) / (\log m_0) \rceil$ unresolved oracle invocations, and there is at most one unresolved read of the seed of $\mathsf{Sim}$.
				\item The seed length and space complexity of $\mathsf{SZA}_m^{\mathsf{Sim}}$ are both $O(s + u \log(w/\epsilon))$.
			\end{enumerate}
		\end{thm}
	\end{samepage}
	
	To illustrate the theorem statement, we demonstrate how to recover the original Saks-Zhou result of \cite{sz99}. Let $\mathsf{Gen}: \{0, 1\}^s \to \{0, 1\}^{m_0}$ be the (non-targeted) efficiently computable $\epsilon$-pseudorandom generator against $\tilde{\mathbf{Q}}_{w, d}^{m_0}$ of \cite[Theorem 3]{inw94} with $m_0 = 2^{\sqrt{\log w}}$, $\epsilon = 1/(6 \cdot 12w)$, and $s \leq O(\log^{3/2} w)$. Let $\mathsf{Sim}$ be the corresponding simulator. Then $\mathsf{SZA}_w^{\mathsf{Sim}}$ is a $(1/6)$-simulator for $\tilde{\mathbf{Q}}_{w, d}^{m'}$ for some $m' \geq w$, and hence it can be used to simulate $\mathbf{BPL}$ (by ensuring that all transitions from the halting configurations are self loops.) The parameter $u$ is $O(\sqrt{\log w})$, and hence the seed length and space usage of $\mathsf{SZA}_w^{\mathsf{Sim}}$ are both $O(\log^{3/2} w)$. By Lemma~\ref{lem:implicit-oracle}, the space needed to simulate $\mathsf{SZA}_w^{\mathsf{Sim}}$ by a non-oracle algorithm is $O(\log^{3/2} w)$. Iterating over all seeds proves $\mathbf{BPL} \subseteq \mathbf{DSPACE}(\log^{3/2} n)$, since the number of configurations of a logspace Turing machine on a length $n$ input is $w \leq \poly(n)$.
	
	The rest of this section is the proof of Theorem~\ref{thm:sza}. All of the ideas in the proof are already present in \cite{sz99} and \cite{arm98}. Our main contributions in this section are the \emph{formulation and statement} of Theorem~\ref{thm:sza}, which enable us to derive the consequence expressed in Theorem~\ref{thm:targeted-prg-to-advice-generator}.
	
	\subsection{Randomness efficient samplers}
	The first step to proving Theorem~\ref{thm:sza} is an observation by Armoni \cite{arm98}. Let $\mathsf{NisGen}$ denote Nisan's generator. Saks and Zhou used a special feature of $\mathsf{NisGen}$. The special feature is that the seed can be split into two parts $x, z$ with $z \leq O(\log(w/\epsilon))$ such that for any particular automaton $Q$, for most values of $x$, $\mathsf{NisGen}(x, \cdot)$ is a good pseudorandom generator for $Q$. (Namely, we can let $x$ be the sequence of hash functions and $z$ be the input to those hash functions.) Armoni observed that \emph{any} pseudorandom generator can be made to have this feature just by precomposing with an \emph{averaging sampler}. We give here the appropriate notion of averaging samplers for $[w]$-valued functions:
	\begin{defn}
		Fix $\mathsf{Samp}: \{0, 1\}^{\ell} \times \{0, 1\}^d \to \{0, 1\}^s$. For a function $f: \{0, 1\}^s \to [w]$, we say that a string $x \in \{0, 1\}^{\ell}$ is \emph{$\delta$-good for $f$} if $f(\mathsf{Samp}(x, U_d)) \sim_{\delta} f(U_s)$. We say that $\mathsf{Samp}$ is an \emph{averaging $(\delta, \gamma)$-sampler for $[w]$-valued functions} if for every $f: \{0, 1\}^s \to [w]$,
		\[
		\Pr_{x \sim U_{\ell}}[x \text{ is $\delta$-good for $f$}] \geq 1 - \gamma.
		\]
	\end{defn}
	We need a space-efficient averaging sampler with good parameters. Armoni used Zuckerman's averaging sampler \cite{zuc97}, but Zuckerman's sampler breaks down for extremely small values of $\delta$. Therefore, to get a slightly more general result, we use the GUV extractor \cite{guv09}, or rather a space-optimized version by Kane, Nelson, and Woodruff \cite{knw08}. It is standard that extractors are good samplers; the following lemma expresses the parameters achieved by the space-optimized GUV extractor when it is viewed as a sampler for $[w]$-valued functions:
	\begin{lem} \label{lem:samp}
		For all $s, w \in \N$ and all $\delta, \gamma > 0$, there is an averaging $(\delta, \gamma)$-sampler for $[w]$-valued functions $\mathsf{Samp}: \{0, 1\}^{\ell} \times \{0, 1\}^d \to \{0, 1\}^s$ with
		\[
		\ell \leq O(s) + \log(w/\gamma)
		\]
		and
		\[
		d \leq O(\log s + \log w + \log(1/\delta) + \log \log(1/\gamma)),
		\]
		where $\mathsf{Samp}(x, y)$ can be computed in $O(s + \log(w/\gamma))$ space.
	\end{lem}
	
	\begin{proof}
		Let $\ell = 2s + 1 + \log(w/\gamma)$. By \cite[Theorem A.14]{knw08}, there is a $(2s, 2\delta/w)$-extractor $\mathsf{Samp}: \{0, 1\}^{\ell} \times \{0, 1\}^d \to \{0, 1\}^s$ with $d \leq O(\log \ell + \log(w/\delta))$, which is
		\[
		O(\log s + \log w + \log(1/\delta) + \log \log (1/\gamma))
		\]
		as claimed, such that $\mathsf{Samp}(x, y)$ can be computed in $O(\ell + \log(w/\delta))$ space, which is $O(s + \log(w/\delta))$ space as claimed.
		
		All that remains is to prove correctness. Fix $f: \{0, 1\}^s \to [w]$. Say $x \in \{0, 1\}^{\ell}$ is \emph{good for $f$ with respect to $z \in [w]$} if
		\[
			|\Pr[f(\mathsf{Samp}(x, U_d)) = z] - \Pr[f(U_s) = z]| \leq 2\delta/w.
		\]
		By \cite[Proposition 2.7]{zuc97} (or rather its proof), for each $z \in [w]$,
		\[
			\Pr_{x \sim U_{\ell}}[x \text{ is good for $f$ with respect to $z$}] \geq 1 - 2^{-\log(w/\gamma)} = 1 - \gamma/w.
		\]
		Therefore, by the union bound over the $w$ different values of $z$, the probability that a uniform random $x$ is good for $f$ with respect to \emph{every} $z \in [w]$ simultaneously is at least $1 - \gamma$. For such an $x$, the $\ell_1$ distance between $f(\mathsf{Samp}(x, U_d))$ and $f(U_s)$ is at most $2\delta$. Total variation distance is half $\ell_1$ distance, so such an $x$ is $\delta$-good for $f$, completing the proof.
	\end{proof}
	
	\subsection{The snap operation}
    At the heart of the SZA transformation is a randomized rounding operation that we will call $\mathsf{Snap}$. This operation slightly perturbs a given automaton with fail state. The basic feature of this perturbation is that if $Q \approx Q'$, then with high probability, $\mathsf{Snap}(Q) = \mathsf{Snap}(Q')$. This phenomenon (which we will make rigorous in Lemma~\ref{lem:snap-coincide}) is reminiscent of ``snapping to a grid'', hence the name.
	
	A \emph{substochastic $d$-matrix} is a square matrix $M$ filled with nonnegative multiples of $2^{-d}$ such that for every $q$, $\sum_r M_{qr} \leq 1$. A $(w, d)$-automaton with fail state $Q$ has a \emph{transition probability matrix} $\mathcal{M}(Q)$, a $w \times w$ substochastic $d$-matrix defined by
	\[
		\mathcal{M}(Q)_{qr} = \Pr_{z \in \{0, 1\}^d} [Q(q; z) = r].
	\]
	Conversely, from a $w \times w$ substochastic $d$-matrix $M$, we define a \emph{canonical automaton with fail state} $\mathcal{Q}(M)$ by identifying $\{0, 1\}^d$ with $[2^d]$ and setting
	\[
		\mathcal{Q}(M)(q; z) = \begin{cases}
			\text{the smallest $r$ such that } z 2^{-d} \leq \sum_{r' = 1}^r M_{qr'} & \text{if such an $r$ exists} \\
			w + 1 & \text{otherwise.}
		\end{cases}
	\]
	\begin{defn}
		For $p \in [0, 1]$ and $\Delta \in \N$, define $\lfloor p \rfloor_{\Delta} = \lfloor 2^{\Delta} p\rfloor 2^{-\Delta}$, i.e. $p$ truncated to $\Delta$ bits after the radix point. Define $\mathsf{Snap}: [0, 1] \times \{0, 1\}^* \to [0, 1]$ by
		\[
			\mathsf{Snap}(p, y) = \lfloor \max\{0, p - (0.y) \cdot 2^{-|y|}\} \rfloor_{|y|},
		\]
		where $0.y$ represents a number in $[0, 1]$ in binary.\footnote{In the notation of \cite{sz99} and \cite{arm98}, $\mathsf{Snap}(p, y) = \lfloor \Sigma_{(0.y) 2^{-|y|}}(p) \rfloor_{|y|}$. In the notation of \cite{ccvm06}, $\mathsf{Snap}(p, y) = \mathcal{R}_{y, |y|}(p)$.} Extend the definition to operate on matrices componentwise: $\mathsf{Snap}(M, y)_{qr} = \mathsf{Snap}(M_{qr}, y)$. Further extend $\mathsf{Snap}$ to operate on automata with fail states by the rule $\mathsf{Snap}(Q, y) = \mathcal{Q}(\mathsf{Snap}(\mathcal{M}(Q), y))$. (The second argument to $\mathsf{Snap}$ should be thought of as random bits.)
	\end{defn}
	Let $\|\cdot\|$ denote the \emph{matrix norm}, i.e. the maximum sum of absolute entries of any row. Define a metric on automata with fail states with the same number of states by setting $\rho(Q, Q') = \|\mathcal{M}(Q) - \mathcal{M}(Q')\|$. The following lemma relates this metric to total variation distance.
    \begin{lem} \label{lem:total-variation-distance}
        Suppose $Q$ is a $(w, d)$-automaton with fail state and $Q'$ is a $(w, d')$-automaton with fail state. Let $\delta$ be the maximum, over all $q \in [w + 1]$, of the total variation distance between $Q(q; U_d)$ and $Q'(q; U_{d'})$. Then $\frac{1}{2} \rho(Q, Q') \leq \delta \leq \rho(Q, Q')$.
    \end{lem}
    \begin{proof}
        For each $q, r \in [w + 1]$, let $\rho_{qr} = \Pr[Q(q; U_d) = r] - \Pr[Q'(q; U_{d'}) = r]$. Then $\rho(Q, Q') = \max_{q \in [w]} \sum_{r \in [w]} |\rho_{qr}|$. Since total variation distance is half $L_1$ distance, $\delta = \frac{1}{2} \max_{q \in [w + 1]} \sum_{r \in [w + 1]} |\rho_{qr}|$. This immediately shows that $\frac{1}{2} \rho(Q, Q') \leq \delta$. For the second inequality, let $q$ be such that $\delta = \frac{1}{2} \sum_{r \in [w + 1]} |\rho_{qr}|$. Since $Q$ and $Q'$ are both automata with fail states, $q$ can be chosen to not be $w + 1$, and hence $\rho(Q, Q') \geq \sum_{r \in [w]} |\rho_{qr}| = 2\delta - |\rho_{q, w + 1}|$. Since $\sum_r \rho_{qr} = 0$, $|\rho_{q, w + 1}| \leq \rho(Q, Q')$, so $\rho(Q, Q') \geq 2\delta - \rho(Q, Q')$. Rearranging completes the proof.
    \end{proof}
    
	\begin{lem} \label{lem:snap-closeness}
		For any $(w, d)$-automaton with fail state $Q$ and any $y \in \{0, 1\}^{\Delta}$, $\rho(Q, \mathsf{Snap}(Q, y)) \leq w 2^{-\Delta + 1}$.
	\end{lem}
	\begin{proof}
		The snap operation perturbs each entry of the $w \times w$ matrix by at most $2^{-\Delta + 1}$.
	\end{proof}
	
	\begin{lem} \label{lem:snap-coincide}
		Fix a $(w, d)$-automaton with fail state $Q$ and let $Y \sim U_{\Delta}$. Then
		\[
			\Pr[\exists Q' \text{ such that } \rho(Q, Q') \leq 2^{-2 \Delta} \text{ and yet } \mathsf{Snap}(Q, Y) \neq \mathsf{Snap}(Q', Y)] \leq w^2 2^{-\Delta + 1}.
		\]
	\end{lem}
	
	\begin{proof}
		Let $E_{qr}$ be the bad event that there exists $p$ such that $|\mathcal{M}_{qr} - p| \leq 2^{-2\Delta}$ and yet $\mathsf{Snap}(\mathcal{M}(Q)_{qr}, Y) \neq \mathsf{Snap}(p, Y)$. For $E_{qr}$ to occur, there must be some $x$ a multiple of $2^{-\Delta}$ such that $\mathcal{M}(Q)_{qr} - (0.Y) \cdot 2^{-\Delta}$ is in $[x - 2^{-2\Delta}, x + 2^{-2\Delta})$. There are only two values of $Y$ that can make this happen, so $\Pr[E_{qr}] \leq 2^{-\Delta + 1}$. The union bound completes the proof, since $\|M\| \geq \max_{q, r} |M_{qr}|$.
	\end{proof}

	\subsection{The construction}
	Recall that $w$ is the number of states (excluding the fail state), $\epsilon$ is the error of $\mathsf{Sim}$, and $s$ is the seed length of $\mathsf{Sim}$. Let $\Delta = \lceil \log(w^2/\epsilon)\rceil$, let $\delta = 2^{-2\Delta - 1}$, and let $\gamma = 2\epsilon/w$. Let $\mathsf{Samp}: \{0, 1\}^{\ell} \times \{0, 1\}^d \to \{0, 1\}^s$ be the averaging $(\delta, \gamma)$-sampler for $[w]$-valued functions of Lemma~\ref{lem:samp}. (This defines the constant $c$; note that Lemma~\ref{lem:samp} ensures $d \leq O(\log(w/\epsilon))$, since the theorem statement assumes $s \leq w$.)
	
	We now define a randomized approximate automaton powering operation $\hat{\mathsf{Pow}}$. For a $(w, d)$-automaton with fail state $Q$ and a string $x \in \{0, 1\}^{\ell}$, we define a $(w, d)$-automaton with fail state $\hat{\mathsf{Pow}}(Q, x)$ by the formula
	\[
		\hat{\mathsf{Pow}}(Q, x)(q; z) = \mathsf{Sim}(Q, q, \mathsf{Samp}(x, z)).
	\]
	Recall that $m_0$ is the number of steps simulated by $\mathsf{Sim}$, and note that for any $Q$, for most $x$, $\hat{\mathsf{Pow}}(Q, x) \approx Q^{m_0}$. The idea of the $\mathsf{SZA}$ transformation is to alternately apply $\hat{\mathsf{Pow}}$ and $\mathsf{Snap}$. The $\mathsf{Snap}$ operation allows us to reuse the randomness of the $\hat{\mathsf{Pow}}$ operation from one application to the next, thereby saving random bits.
	
	Let $Q_0$ be the $(w, d)$-automaton with fail state that is given to $\mathsf{SZA}$ as input. Recall that $u = \lceil (\log m) / (\log m_0) \rceil$, where $m$ is the number of steps of $Q_0$ that $\mathsf{SZA}$ is trying to simulate. For a sequence $y = (y_1, \dots, y_u) \in \{0, 1\}^{\Delta u}$ and a string $x \in \{0, 1\}^{\ell}$, we define a sequence of $(w, d)$-automata with fail states $\hat{Q}_0[x, y], \dots, \hat{Q}_u[x, y]$ by starting with $\hat{Q}_0[x, y] = Q_0$ and setting
	\[
		\hat{Q}_{i + 1}[x, y] = \mathsf{Snap}(\hat{\mathsf{Pow}}(\hat{Q}_i[x, y], x), y_{i + 1}).
	\]
	(For $i \geq 1$, $Q_i$ is naturally thought of as a $(w, \Delta)$-automaton with fail state, but since $\Delta \leq d$, we can think of it as reading $d$ bits for each transition and ignoring all but the first $\Delta$ of them.) Finally, for seed values $x \in \{0, 1\}^{\ell}, y \in \{0, 1\}^{\Delta u}, z \in \{0, 1\}^d$, we set \[
		\mathsf{SZA}_m^{\mathsf{Sim}}(Q_0, q, x, y, z) := \hat{Q}_u[x, y](q; z).
	\]
	
	\subsection{Correctness}
	The bulk of the correctness proof consists of justifying the fact that we use the same $x$ value for each application of $\hat{\mathsf{Pow}}$ in the definition of $\hat{Q}_i$. To do this, we define a \emph{deterministic} approximate powering operation $\mathsf{Pow}$. For a $(w, d)$-automaton with fail state $Q$, define a $(w, s)$-automaton with fail state $\mathsf{Pow}(Q)$ by
	\[
		\mathsf{Pow}(Q)(q; z) = \mathsf{Sim}(Q, q, z).
	\]
	Note that $\mathsf{Pow}(Q) \approx Q^{m_0}$. For a sequence $y = (y_1, \dots, y_u) \in \{0, 1\}^{\Delta u}$, define (just for the analysis) another sequence of $(w, d)$-automata with fail states $Q_0[y], \dots, Q_u[y]$ by starting with $Q_0[y] = Q_0$ and setting
	\[
		Q_{i + 1}[y] = \mathsf{Snap}(\mathsf{Pow}(Q_i[y]), y_{i + 1}).
	\]
	We first verify that these automata $Q_i$ (always) provide good approximations for the true powers of $Q_0$:
	
	\begin{lem} \label{lem:sza-correctness-closeness}
		For any $y$, $\rho(Q_u[y], Q_0^{m_0^u}) \leq 8m\epsilon$.
	\end{lem}
	
	\begin{proof}
		We show by induction on $i$ that
		\[
			\rho(Q_i[y], Q_0^{m_0^i}) \leq \frac{m_0^i - 1}{m_0 - 1} \cdot (2\epsilon + w 2^{-\Delta + 1}).
		\]
		In the base case $i = 0$, this is immediate. For the inductive step, by the triangle inequality,
		\[
			\rho(Q_{i + 1}[y], Q_0^{m_0^{i + 1}}) \leq \rho(Q_{i + 1}[y], \mathsf{Pow}(Q_i[y])) + \rho(\mathsf{Pow}(Q_i[y]), Q_i[y]^{m_0}) + \rho(Q_i[y]^{m_0}, Q_0^{m_0^{i + 1}}).
		\]
		The first term is at most $w 2^{-\Delta + 1}$ by Lemma~\ref{lem:snap-closeness}. The second term is at most $2\epsilon$ by the simulator guarantee and Lemma~\ref{lem:total-variation-distance}. The third term is at most $m_0 \rho(Q_i[y], Q^{m_0^i})$ by \cite[Proposition 2.3]{sz99}. Therefore, by the inductive assumption,
		\begin{align*}
			\rho(Q_{i + 1}[y], Q_0^{m_0^{i + 1}}) &\leq w 2^{-\Delta + 1} + 2\epsilon + m_0 \cdot \frac{m_0^i - 1}{m_0 - 1} \cdot (2\epsilon + w 2^{-\Delta + 1}) \\
			&= \frac{m_0^{i + 1} - 1}{m_0 - 1} \cdot (2\epsilon + w 2^{-\Delta + 1}).
		\end{align*}
		That completes the induction. Finally, we plug in $i = u$:
		\[
			\rho(Q_u[y], Q_0^{m_0^u}) \leq \frac{m_0^u - 1}{m_0 - 1}(2\epsilon + 2w 2^{-\Delta}) \leq 2m \cdot (2\epsilon + 2\epsilon). \qedhere
		\]
	\end{proof}
	
	Now, we show that the $\mathsf{Snap}$ operation ensures that with high probability, $\hat{Q}_i$ and $Q_i$ are exactly equal, despite their different definitions:
	
	\begin{lem} \label{lem:sza-correctness-snapping}
		Let $X \sim U_{\ell}, Y_1 \sim U_{\Delta}, \dots, Y_u \sim U_{\Delta}$ all be independent. Then
		\[
			\Pr[\text{there is some $i \leq u$ such that } \hat{Q}_i[X, Y] \neq Q_i[Y]] \leq 4m\epsilon.
		\]
	\end{lem}
	
	\begin{proof}
		By the sampling property, Lemma~\ref{lem:total-variation-distance}, and a union bound over the $w$ different start states, for each $i \in \{0, \dots, u - 1\}$,
		\begin{equation} \label{eqn:bad-event-1}
			\Pr[\rho(\mathsf{Pow}(Q_i[Y]), \hat{\mathsf{Pow}}(Q_i[Y], X)) > 2\delta] \leq w \gamma = 2\epsilon.
		\end{equation}
		(Imagine picking $Y$ first and then taking a probability over the randomness of $X$ alone.) Now, $2\delta = 2^{-2\Delta}$, and by Lemma~\ref{lem:snap-coincide},
		\begin{align}
			\Pr\begin{bmatrix}\exists Q' \text{ such that } \rho(\mathsf{Pow}(Q_i[Y]), Q') \leq 2^{-2\Delta} \\ \text{ and } \mathsf{Snap}(\mathsf{Pow}(Q_i[Y]), Y_{i + 1}) \neq \mathsf{Snap}(Q', Y_{i + 1})\end{bmatrix} &\leq w^2 2^{-\Delta + 1} \label{eqn:bad-event-2} \\
			&\leq 2\epsilon.
		\end{align}
		By the union bound over the $u$ different values of $i$, the probability that \emph{any} of these bad events occur is at most $u(2\epsilon + 2\epsilon) \leq 4m\epsilon$. So to prove the lemma, assume that \emph{none} of these bad events occur. In this case, we show by induction that $\hat{Q}_i[X, Y] = Q_i[Y]$ for every $0 \leq i \leq u$. The base case $i = 0$ holds by definition. For the inductive step, assume $\hat{Q}_i[X, Y] = Q_i[Y]$. Then because we assumed that the bad event of Equation~\ref{eqn:bad-event-1} did not occur, $\rho(\hat{\mathsf{Pow}}(\hat{Q}_i[X, Y], X), \mathsf{Pow}(\hat{Q}_i[Y])) \leq 2^{-2\Delta}$. And hence because we assumed that the bad event of Equation~\ref{eqn:bad-event-2} also did not occur, we may conclude that
		\[
			\mathsf{Snap}(\hat{\mathsf{Pow}}(\hat{Q}_i[X, Y], X), Y_{i + 1}) = \mathsf{Snap}(\mathsf{Pow}(Q_i[Y]), Y_{i + 1}).
		\]
		By definition, this implies that $\hat{Q}_{i + 1}[X, Y] = Q_{i + 1}[Y]$.
	\end{proof}
	
	We have shown that $Q_1, Q_2, \dots, Q_u$ provide good approximations of true powers of $Q_0$, and with high probability, $\hat{Q}_i = Q_i$ for every $i$. It immediately follows that a random transition of $\hat{Q}_u$ gives a similar distribution as $m_0^u$ random transitions of $Q_0$:
	
	\begin{proof}[Proof of correctness of $\mathsf{SZA}$]
        Lemmas~\ref{lem:sza-correctness-closeness} and~\ref{lem:sza-correctness-snapping} imply that
        \[
            \Pr[\rho(\hat{Q}_u[X, Y], Q_0^{m_0^u}) \leq 8m\epsilon] \geq 1 - 4m\epsilon.
        \]
        By Lemma~\ref{lem:total-variation-distance}, if $x$ and $y$ are such that $\rho(\hat{Q}_u[x, y], Q_0^{m_0^u}) \leq 8m\epsilon$, then $\hat{Q}_u[x, y](q; Z) \sim_{8m\epsilon} Q_0^{m_0^u}(q; U_{dm_0^u})$. An averaging argument completes the proof.
	\end{proof}
	
	\subsection{Efficiency} \label{sec:sza-efficiency}
	The seed length of $\mathsf{SZA}$ is $\ell + u\Delta + d$, which is $O(s + u \log(w/\epsilon))$. We argue that $\mathsf{SZA}$ can be implemented to run in $O(s + u \log(w/\epsilon))$ space through mutual recursion involving two subroutines. The first subroutine, given $i, r, z'$, computes $\hat{Q}_i[x, y](r; z')$:
	\begin{enumerate}
		\item If $i = 0$, just consult the input directly. Otherwise:
		\item Use the second subroutine to obtain each required entry of $\mathcal{M}(\hat{\mathsf{Pow}}(\hat{Q}_{i - 1}[x, y], x))$. Apply the definition of $\hat{Q}_i$ directly.
	\end{enumerate}
	The space used by this subroutine is only $O(\log(w/\epsilon))$ plus the space required for computing each matrix entry. The second subroutine, given $i, r, v$, computes $\mathcal{M}(\hat{\mathsf{Pow}}(\hat{Q}_i[x, y], x))_{rv}$:
	\begin{enumerate}
		\item Initialize $\xi = 0$. For all $z' \in \{0, 1\}^d$:
		\begin{enumerate}
			\item Use the oracle to compute $\hat{\mathsf{Pow}}(\hat{Q}_i[x, y], x)(r; z')$. If it gives $v$, set $\xi := \xi + 2^{-d}$. When the oracle makes reads to its automaton/start state inputs, use the first subroutine to compute the necessary values of $\hat{Q}_i[x, y]$. When the oracle makes reads to its seed inputs, (re)compute $\mathsf{Samp}(x, z')$ to obtain the appropriate bit.
		\end{enumerate}
		\item Output $\xi$.
	\end{enumerate}
	This subroutine's space usage can get up to $O(s + \log(w/\epsilon))$ for computing the sampler, but before each recursive call, it erases all but $O(\log(w/\epsilon))$ bits. By induction, this shows that the total space usage of each of these two subroutines (including now the space used for recursive calls) is $O(s + (i + 1) \log(w/\epsilon))$. It follows that the space used by $\mathsf{SZA}$ is $O(s + u \log(w/\epsilon))$, since it just requires a call to the first subroutine with $i = u$.
	
	In this implementation, the maximum number of unresolved oracle invocations at any time is indeed $u$, and there is indeed at most one unresolved read of a seed. This completes the proof of Theorem~\ref{thm:sza}. \qed
	
	\section{Transforming simulators into targeted PRGs} \label{sec:simulator-to-targeted-prg}
	
	Recall from Section~\ref{sec:proof-overview} that to prove the harder direction of our main result, we require three transformations: an assumed transformation of targeted pseudorandom generators into simulation advice generators, the SZA transformation, and a transformation of simulators into targeted pseudorandom generators. In this section, we construct the last transformation.
	
	We state our transformation in terms of the Ladner-Lynch (LL) oracle model \cite{ll76}. This model is simpler than the implicit oracle model of Section~\ref{sec:implicit-oracle}. An LL-model oracle algorithm has a single write-only oracle tape. When the algorithm makes a query, the contents of the oracle tape are erased, and the answer to the query is stored in the algorithm's state. Symbols written on the oracle tape do not count toward the algorithm's space complexity. For a non-Boolean oracle $f: \{0, 1\}^* \to \{0, 1\}^*$, the oracle algorithm is required to specify an index $i$ along with the query string $x$; the oracle responds with $f(x)_i$. We emphasize that as with the SZA transformation, this oracle model is only used to cleanly express the transformation; ultimately, we will plug in actual algorithms in place of the oracle.
	
	\begin{lem} \label{lem:simulator-to-targeted-prg}
		There exists a deterministic LL-model oracle algorithm $\mathsf{G}$ such that if $\mathsf{Sim}$ is an $\epsilon$-simulator for $\mathbf{Q}_{wm, d}^m$ with seed length $s$, then:
		\begin{enumerate}
			\item $\mathsf{G}^{\mathsf{Sim}}$ is a targeted $(2mw^2 \epsilon)$-pseudorandom generator against $\mathbf{Q}_{w, d}^m$.
			\item $\mathsf{G}^{\mathsf{Sim}}$ has seed length $s$ and space complexity $O(s + d + \log w + \log m)$.
		\end{enumerate}
	\end{lem}
	
    To prove Lemma~\ref{lem:simulator-to-targeted-prg}, we use $\mathsf{Sim}$ to choose a final state, and then we use $\mathsf{Sim}$ to ``reverse engineer'' a string that brings $Q$ to that final state. This reverse engineering process is a straightforward application of the method of conditional probabilities.
    
	\begin{proof}
		Given $(Q, q, x)$:
		\begin{enumerate}
            \item Let $Q'$ be the $(wm, d)$-automaton formed by adding dummy states to $Q$. Use the oracle to set $R := \mathsf{Sim}(Q', q, x)$.
			\item Initialize $v = q$. For $i = 0$ to $m - 1$:
			\begin{enumerate}
				\item For each $z \in \{0, 1\}^d$, let $v_z = Q(v; z)$.
				\item Let $Q'$ be a $(wm, d)$-automaton that simulates $m - i$ steps of $Q$, with $v'_z$ being the start state corresponding to $v_z$ and $R'$ being the end state corresponding to $R$.
				\item Compute the $z \in \{0, 1\}^d$ that maximizes $\#\{x' : \mathsf{Sim}(Q', v'_z, x') = R'\}$, breaking ties arbitrarily.
				\item Print $z$ and set $v := v_z$.
			\end{enumerate}
		\end{enumerate}
		Clearly, $\mathsf{G}$ outputs $dm$ bits and uses $O(s + d + \log w + \log m)$ space. Proof of correctness: For $t, r \in [w]$ and $i \in \{0, \dots, m - 1\}$, let $p_{t, r}[i] = \Pr[Q^{m - i}(t; U_{m - i}) = r]$. We show by induction on $i$ that at the beginning of iteration $i$ of the loop on line $2$, $p_{v, R}[i] \geq p_{q, R}[0] - 2i\epsilon$. Base case: At the beginning of iteration $i = 0$, $v = q$. Inductive step: Consider the execution of iteration $i$ of the loop. By the simulator guarantee, there is some $z \in \{0, 1\}^d$ such that $\#\{x': \mathsf{Sim}(Q', v'_z, x') = R\} \geq (p_{v, R}[i] - \epsilon) 2^s$. Therefore, $\mathsf{G}$ chooses a $z$ that also satisfies that inequality. Therefore, applying the simulator guarantee again, $p_{v_z, R}[i + 1] \geq p_{v, R}[i] - 2\epsilon$. This completes the induction.
		
		Now, let $X \sim U_s$, and let $Y = \mathsf{G}^{\mathsf{Sim}}(Q, q, X)$. Fix an arbitrary state $r \in [w]$; we will show that $\Pr[Q^m(q; Y) = r]$ is close to $\Pr[Q^m(q; U_{dm}) = r]$. Say $r$ is \emph{typical} if $p_{q, r}[0] \geq 2m\epsilon$. For the first case, suppose $r$ is typical. By the fact that we proved by induction, $\Pr[Q^m(q; Y) = R \mid R \text{ is typical}] = 1$. Therefore,
        \begin{align*}
			\Pr[Q^m(q; Y) = r] &= \Pr[Q^m(q; Y) = r \mid R = r] \cdot \Pr[R = r] + \Pr[Q^m(q; Y) = r \mid R \neq r] \cdot \Pr[R \neq r] \\
			&= \Pr[R = r] + \Pr[Q^m(q; Y) = r \mid R \neq r] \cdot \Pr[R \neq r].
		\end{align*}
        This expression is \emph{lower} bounded by $\Pr[R = r]$, which is lower bounded by $\Pr[Q^m(q; U_{dm}) = r] - \epsilon$ by the simulator guarantee. On the other hand, the expression is \emph{upper} bounded by $\Pr[R = r] + \Pr[R \text{ is atypical}]$, which is upper bounded by $\Pr[Q^m(q; U_{dm}) = r] + \epsilon + 2mw\epsilon$ by the simulator guarantee, the definition of typicality, and the union bound.
        
        For the second case, suppose $r$ is atypical. Then $Q^m(q; Y) = r$ implies that $R$ is atypical, which happens with probability at most $2mw\epsilon + \epsilon$ by the definition of typicality and the simulator guarantee.
        
        Therefore, in either case, $\Pr[Q^m(q; Y) = r]$ is within $\pm (2mw + 1)\epsilon$ of $\Pr[Q^m(q; U_{dm}) = r]$. Statistical distance is half $L_1$ distance, so the error of $\mathsf{G}^{\mathsf{Sim}}$ is at most $\frac{1}{2} w(2mw + 1)\epsilon \leq 2mw^2 \epsilon$.
	\end{proof}
	
	\section{Proof of Theorem~\ref{thm:targeted-prg-to-advice-generator}}
	\subsection{Composing the transformations}
	
	In this section, we compose the transformation of Condition~$\ref{cond:targeted-prg-to-advice-generator}$ of Theorem~\ref{thm:targeted-prg-to-advice-generator}, the SZA transformation, and the transformation of Section~\ref{sec:simulator-to-targeted-prg}. (In the overview of Section~\ref{sec:proof-overview}, this corresponds to the composition of steps 1, 2, and 3.) The composition is a transformation on targeted pseudorandom generators:
	
	\begin{lem} \label{lem:cycle}
		Assume Condition~$\ref{cond:targeted-prg-to-advice-generator}$ of Theorem~$\ref{thm:targeted-prg-to-advice-generator}$ is true. Fix a constant $\beta > 0$, sufficiently small constants $\sigma > \eta > \gamma > 0$, and a constant $\mu \in (\gamma, 1 - \beta]$. Suppose there is a family $\{\mathsf{Gen}_w\}$, where $\mathsf{Gen}_w$ is an efficiently computable targeted $\epsilon$-pseudorandom generator against $\mathbf{Q}_{w, 1}^m$ with seed length $s$ satisfying
		\begin{align*}
			&s \leq O(\log^{1 + \sigma} w),
			&&\log(1/\epsilon) = \log^{1 + \eta} w,
			&&&\log m \geq \log^{\mu} w. \\
			\intertext{Then there is another family $\{\mathsf{Gen}'_w\}$, where $\mathsf{Gen}'_w$ is an efficiently computable targeted $\epsilon'$-pseudorandom generator against $\mathbf{Q}_{w, 1}^{m'}$ with seed length $s'$ satisfying}
			&s' \leq O(\log^{1 + \max\{\sigma, \beta\} + 4\eta} w),
			&&\log(1/\epsilon') \geq \Omega(\log^{1 + \eta - \gamma} w),
			&&&\log m' \geq \log^{\mu + \beta} w.
		\end{align*}
	\end{lem}
	
	All the hard work of proving Lemma~\ref{lem:cycle} has already been done in Sections~\ref{sec:saks-zhou} and~\ref{sec:simulator-to-targeted-prg}; conceptually, the proof is simply by composing. Some technicalities complicate matters slightly. First, we need two little lemmas to deal with the fact that $d > 1$ in Theorem~\ref{thm:sza}, to deal with the fact that Theorem~\ref{thm:sza} is phrased in terms of automata with fail states, and to deal with the relationship between $w$ and $m$ in Lemma~\ref{lem:simulator-to-targeted-prg}.
	
	\begin{lem} \label{lem:advice-generator-d}
		Suppose $\mathsf{Gen}$ is an $\epsilon$-simulation advice generator for $\mathbf{Q}_{(w + 1)2^d, 1}^{m}$. Then $\mathsf{Gen}$ is also an $\epsilon$-simulation advice generator for $\tilde{\mathbf{Q}}_{w, d}^{\lfloor m/d \rfloor}$.
	\end{lem}
	
	\begin{proof}
		Let $\mathsf{S}$ be the logspace algorithm such that $\mathsf{S}(Q, q, \mathsf{Gen}(x))$ is an $\epsilon$-simulator for $\mathbf{Q}_{(w + 1)2^d, 1}^{m}$. Let $a$ be the output length of $\mathsf{Gen}$. For a $(w, d)$-automaton with fail state $Q$, a start state $q \in [w + 1]$, and a string $y \in \{0, 1\}^a$, let $\mathsf{S'}(Q, q, y)$ behave as follows:
		\begin{enumerate}
			\item Let $Q'$ be the $((w + 1)2^d, 1)$-automaton that simulates $Q$. (One step of $Q$ is simulated by $d$ steps of $Q'$; the state space of $Q'$ is $[w + 1] \times \{0, 1\}^{< d}$.) Let $q'$ be the start state of $Q'$ corresponding to $q$.
			\item Let $r' = \mathsf{S}(Q', q', y)$.
			\item Return the state $r \in [w + 1]$ that corresponds to $r'$.
		\end{enumerate}
		The maps $(Q, q, y) \mapsto (Q', q', y)$ and $r' \mapsto r$ are computable in logspace, so $\mathsf{S'}$ can be implemented to run in logspace. Clearly, $\mathsf{S'}(Q, q, \mathsf{Gen}(x))$ is an $\epsilon$-simulator for $\tilde{\mathbf{Q}}_{w, d}^{\lfloor m/d \rfloor}$.
	\end{proof}
	
	\begin{lem} \label{lem:exact-simulation-length}
		There exists a deterministic LL-model oracle algorithm $\mathsf{R}$ with the following properties. Pick $m \leq m'$ and $d \leq d'$. Suppose $\mathsf{Sim}$ is an $\epsilon$-simulator for $\tilde{\mathbf{Q}}_{wm(m + 1), d'}^{m'}$ with seed length $s$. Then $\mathsf{R}^{\mathsf{Sim}}_{m, d}$ is an $\epsilon$-simulator for $\mathbf{Q}_{wm, d}^m$ with seed length $s$. (Here $m, d$ are inputs to $\mathsf{R}$; we write them as subscripts merely to separate them from the usual simulator inputs.) Further, $\mathsf{R}^{\mathsf{Sim}}_{m, d}$ only uses space $O(d' + \log w + \log m)$.
	\end{lem}
	
	\begin{proof}
		Given $(Q, q, x)$ and oracle access to $\mathsf{Sim}$:
		\begin{enumerate}
			\item Let $Q'$ be a $(wm(m + 1), d')$-automaton with fail state on state space $[wm] \times [m + 1]$ (plus a fail state) defined by
			\[
			Q'((q, t); y) =
			\begin{cases}
			(Q(q; y \restriction_d), t + 1) & \text{if } t \leq m \\
			(q, t) & \text{if } t = m + 1.
			\end{cases}
			\]
			Here $y \restriction_d$ denotes the first $d$ bits of $y$. 
			\item Output the first coordinate of $\mathsf{Sim}(Q', (q, 1), x)$.
		\end{enumerate}
		The first coordinate of $(Q')^{m'}((q, 1); U_{m' d'})$ is distributed identically to $Q^m(q; U_{md})$, and applying a deterministic function (such as ``the first coordinate of'') can only make distributions closer, so this algorithm is correct. Clearly, $Q'$ can be computed from $Q$ in space $O(d' + \log w + \log m)$.
	\end{proof}
	
	Now we are ready to prove Lemma~\ref{lem:cycle}; the proof mainly consists in verifying parameters.
	
	\begin{proof}[Proof of Lemma~$\ref{lem:cycle}$]
		Using Condition~\ref{cond:targeted-prg-to-advice-generator} of Theorem~\ref{thm:targeted-prg-to-advice-generator}, transform the family $\{\mathsf{Gen}_w\}$ into a family $\{\mathsf{AdvGen}_w\}$ of simulation advice generators. For each $w$, let $\mathsf{Sim}_w$ be the simulator induced by $\mathsf{AdvGen}_{(w + 1)2^d}$ using Lemma~\ref{lem:advice-generator-d}, where $d = \lceil c [\log^{1 + \eta - \gamma}(w) + \log(w)] \rceil$ and $c$ is the constant in Theorem~\ref{thm:sza}. Define
		\[
			\mathsf{Sim}'_w = \mathsf{SZA}_{m'}^{\mathsf{Sim}_w} \quad \text{where } \log m' = \lceil \log^{\mu + \beta} w \rceil.
		\]
		Define
		\[
			\mathsf{Sim}''_w = \mathsf{R}_{m', 1}^{\mathsf{Sim}'_{wm'(m' + 1)}},
		\]
		where $\mathsf{R}$ is the algorithm of Lemma~\ref{lem:exact-simulation-length}. Finally, define
		\[
			\mathsf{Gen}'_w = \mathsf{G}^{\mathsf{Sim}''_w},
		\]
		where $\mathsf{G}$ is the algorithm of Lemma~\ref{lem:simulator-to-targeted-prg}.
		
		Now that we have constructed $\mathsf{Gen}'_w$, we show that our construction worked. Since $\log^{1 + \eta - \gamma} w$ is monotone increasing, $\mathsf{Gen}'_{(w + 1)2^d}$ can be thought of as having error $\epsilon_0$ where $\log(1/\epsilon_0) = \log^{1 + \eta - \gamma} w$. Therefore, $\mathsf{Sim}_w$ is an $\epsilon_0$-simulator for $\tilde{\mathbf{Q}}_{w, d}^{m_0}$, where $\log m_0 \geq \Omega(\log^{\mu - \gamma}(w) - \log(d)) = \Omega(\log^{\mu - \gamma} w)$. Observe that the chosen $d$ value is exactly $\lceil c \log(w/\epsilon_0) \rceil$. Therefore, by Theorem~\ref{thm:sza}, $\mathsf{Sim}'_w$ is a $(12w\epsilon_0)$-simulator for $\tilde{\mathbf{Q}}_{w, d}^{m_1}$ for some $m_1 \geq m'$. Again using monotonicity, we can think of $\mathsf{Sim}'_{wm'(m' + 1)}$ as having the same error. By Lemma~\ref{lem:exact-simulation-length}, this implies that $\mathsf{Sim}''_w$ is a $(12 w \epsilon_0)$-simulator for $\mathbf{Q}_{wm', 1}^{m'}$, and hence $\mathsf{Gen}'_w$ is a targeted $\epsilon'$-pseduorandom generator against $\mathbf{Q}_{w, 1}^{m'}$, where $\epsilon' = 24 m' w^3 \epsilon_0$, and hence $\log(1/\epsilon') \geq \Omega(\log^{1 + \eta - \gamma} w)$ as desired.
		
		The seed length of $\mathsf{Sim}_w$ is $s_0 \leq O(\log^{1 + \sigma + \gamma}((w + 1)2^d))$, which is $O(\log^{(1 + \sigma + \gamma)(1 + \eta - \gamma)} w)$. Since $1 + \sigma + \eta + \gamma + \sigma \eta + \gamma \eta < 1 + \sigma + 4\eta$, we have $s_0 \leq O(\log^{1 + \sigma + 4\eta} w)$. The parameter $u$ of Theorem~\ref{thm:sza} is bounded by
		\[
			u \leq O\left(\frac{\log^{\mu + \beta} w}{\log^{\mu - \gamma} w}\right) = O(\log^{\beta + \gamma} w).
		\]
		Therefore, the seed length of $\mathsf{Sim}'_w$ is $O(s_0 + u \log(w/\epsilon_0)$, which is $O(\log^{1 + \sigma + 4\eta}(w) + \log^{1 + \beta + \eta}(w))$, which is $O(\log^{1 + \max\{\sigma, \beta\} + 4\eta} w)$. Thus the seed length of $\mathsf{Sim}''_w$ is $O(\log^{1 + \max\{\sigma, \beta\} + 4\eta} \poly(w))$, which is $O(\log^{1 + \max\{\sigma, \beta\} + 4\eta} w)$. Hence the seed length of $\mathsf{Gen}'_w$ is the same.
		
		The output length $a$ of $\mathsf{AdvGen}_{w 2^d}$ satisfies $\log a \leq O(\log^{1 + \eta + \gamma}((w + 1)2^d))$, which is $O(\log^{(1 + \eta + \gamma)(1 + \eta - \gamma)} w)$. Since $(1 + \eta)^2 \leq 1 + 3\eta$, we have $a \leq O(\log^{1 + 3\eta} w)$. Therefore, by Lemma~\ref{lem:implicit-oracle} and Theorem~\ref{thm:sza}, the space complexity of $\mathsf{Sim}'_w$ is $O(s_0 + u[d + \log w + \log a])$, which is $O(\log^{1 + \sigma + 4\eta}(w) + \log^{1 + \beta + 3\eta + \gamma} w)$, which is $O(\log^{1 + \max\{\sigma, \beta\} + 4\eta} w)$. Therefore, by Lemma~\ref{lem:exact-simulation-length}, the space complexity of $\mathsf{Sim}''_w$ satisfies the same bound, and hence so does that of $\mathsf{Gen}'_w$.
	\end{proof}

	\subsection{Iterating the composition}
    
    In this section, we prove the $(2 \implies 1)$ direction of Theorem~\ref{thm:targeted-prg-to-advice-generator}, i.e. we give a strong derandomization under the assumption that targeted pseudorandom generators can be transformed into simulation advice generators. The proof follows the idea outlined in Section~\ref{sec:proof-overview}: we repeatedly apply the composition transformation of the last section $t$ times for an arbitrarily large constant $t$. Each application substantially increases the output length of our targeted pseudorandom generator while the other parameters degrade negligibly, so we end up with an efficiently computable targeted pseudorandom generator with output length $w$ and seed length $O(\log^{1 + O(1/t)} w)$:
    
    \begin{lem} \label{lem:iterate}
        Assume Condition~$\ref{cond:targeted-prg-to-advice-generator}$ of Theorem~$\ref{thm:targeted-prg-to-advice-generator}$ is true. Fix $\alpha > 0$. There is a family $\{\mathsf{Gen}_w\}$, where $\mathsf{Gen}_w$ is an efficiently computable targeted $(1/6)$-pseudorandom generator against $\mathbf{Q}_{w, 1}^m$ with seed length $O(\log^{1 + \alpha} w)$ where $m \geq w$.
    \end{lem}
    
    \begin{proof}
        Let $t = \lceil 2/\alpha \rceil$, $\beta = 1/t$, $\eta = \alpha / (8t)$, and $\gamma = \eta/(3t)$. We show by induction that for $1 \leq i \leq t$, there is a family $\{\mathsf{Gen}_w\}$, where $\mathsf{Gen}_w$ is an efficiently computable targeted $\epsilon_i$-pseudorandom generator against $\mathbf{Q}_{w, 1}^{m_i}$ with seed length $s_i$ satisfying
        \begin{align*}
        	&s_i \leq O(\log^{1 + \beta + 4i\eta} w), &&\log(1/\epsilon_i) = \log^{1 + \eta - 2i\gamma} w, &&&\log m_i \geq \log^{i \beta} w.
        \end{align*}
        For the base case $i = 1$, use the generator of \cite[Theorem 3]{inw94}. For the chosen output length and error, the seed length is $O((\log^{1 + \eta - 2\gamma} w)(\log^{\beta} w))$. For the inductive step, suppose we have constructed family $i$. Apply Lemma~\ref{lem:cycle} to this family, using the chosen $\beta, \gamma$ values. (By the choice of $\gamma$, $\eta - 2i\gamma > 0$.) The parameters of the resulting family are all correct except that the error merely satisfies $\log(1/\epsilon') \geq \Omega(\log^{1 + \eta - (2i + 1) \gamma} w)$; for sufficiently large $w$, this is at least $\log^{1 + \eta - 2(i + 1)\gamma} w$, so modifying finitely many elements of the family gives family $i + 1$.
        
        That completes the induction. To prove the lemma, use family $t$. The output length is at least $w$ as desired, and the error is subconstant as desired. The space complexity and seed length are $\log^{1 + \beta + 4t\eta} w$. By the choices of $\beta, \eta, \gamma$, $1 + \beta + 4t\eta \leq 1 + \alpha$ as desired (as long as $t$ is sufficiently large.)
    \end{proof}
    
    \begin{proof}[Proof of the $(\ref{cond:targeted-prg-to-advice-generator} \implies \ref{cond:derand})$ direction of Theorem~$\ref{thm:targeted-prg-to-advice-generator}$]
        Fix some promise problem
        \[
	        A \in \bigcap_{\beta > 0} \mathbf{promise\mhyphen BPSPACE}(\log^{1 + \beta} n)
	    \]
	    and a constant $\alpha > 0$. Let $M$ be a probabilistic space-$O(\log^{1 + \alpha} n)$ Turing machine that decides $A$ with error $1/6$. Without loss of generality, assume that $M$ has unique accept/reject configurations. On input $x \in \{0, 1\}^n$:
        \begin{enumerate}
        	\item Let $Q$ be a $(w, 1)$-automaton corresponding to the execution of $M(x)$: each state of $Q$ specifies tape contents and a read head location of $M$, and the transitions of $Q$ correspond to $M$ reading a single random bit. Let the transitions from the accept/reject configurations be self-loops.
        	\item Use the generator of Lemma~\ref{lem:iterate} (with the chosen $\alpha$ value) to deterministically simulate $Q$ by iterating over all seeds and taking a majority vote. Accept or reject accordingly.
        \end{enumerate}
        The value $w$ satisfies $\log w \leq O(\log^{1 + \alpha} n)$, and $Q$ can be produced from $x$ in deterministic space $O(\log^{1 + \alpha} n)$. The space needed for the simulation is $O(\log^{1 + \alpha} w)$, which is $O(\log^{(1 + \alpha)^2} n)$. Therefore, the composition algorithm deterministically decides $A$ in space $O(\log^{1 + 2\alpha + \alpha^2} n)$. Since $\alpha$ was arbitrary and $\lim_{\alpha \to 0} (1 + 2\alpha + \alpha^2) = 1$, this shows that $A \in \bigcap_{\alpha > 0} \mathbf{promise\mhyphen DSPACE}(\log^{1 + \alpha} n)$.
    \end{proof}
	
	\subsection{Transforming targeted PRGs into advice generators, assuming derandomization}
	
	In this section, we finally prove the easier half of Theorem~\ref{thm:targeted-prg-to-advice-generator}, i.e. we prove that targeted pseudorandom generators can be transformed to simulation advice generators under strong derandomization assumptions. This is essentially immediate from the definitions: under strong derandomization assumptions, no advice is needed to simulate automata, so the identity function (padded appropriately) is trivially a simulation advice generator.
	
	\begin{lem} \label{lem:targeted-prg-to-advice-generator-easy-direction}
		If $\mathbf{promise\mhyphen BPL} \subseteq \bigcap_{\alpha > 0} \mathbf{promise\mhyphen DSPACE}(\log^{1 + \alpha} n)$, then for any $\eta, \gamma > 0$, there is a family $\{\mathsf{Gen}_w\}$, where $\mathsf{Gen}_w$ is an efficiently computable $\epsilon$-simulation advice generator for $\mathbf{Q}_{w, 1}^w$ with seed length $s$ satisfying
		\begin{align*}
			&s \leq O(\log^{1 + \eta} w)
			&&\log(1/\epsilon) \geq \Omega(\log^{1 + \eta} w)
			&&&\log a \leq O(\log^{1 + \eta + \gamma} w).
		\end{align*}
	\end{lem}
	
	\begin{remark}
		For the purpose of proving Theorem~\ref{thm:targeted-prg-to-advice-generator}, Lemma~\ref{lem:targeted-prg-to-advice-generator-easy-direction} only needed to conclude with Condition~\ref{cond:targeted-prg-to-advice-generator} of the theorem, i.e. a \emph{transformation} from targeted pseudorandom generators to simulation advice generators. But it turns out that under the derandomization assumption of Lemma~\ref{lem:targeted-prg-to-advice-generator-easy-direction}, we can just construct a simulation advice generator ``from scratch.''
	\end{remark}
	
	\begin{remark}
		The derandomization premise of Lemma~\ref{lem:targeted-prg-to-advice-generator-easy-direction} may seem weaker than the derandomization statement in Theorem~\ref{thm:targeted-prg-to-advice-generator} (since it is about $\mathbf{promise\mhyphen BPL}$ instead of $\bigcap_{\alpha > 0} \mathbf{promise\mhyphen BPSPACE}(\log^{1 + \alpha} n)$). This would again make Lemma~\ref{lem:targeted-prg-to-advice-generator-easy-direction} stronger than necessary. But the two derandomization statements are actually equivalent by a padding argument.
	\end{remark}
	
	\begin{proof}[Proof of Lemma~$\ref{lem:targeted-prg-to-advice-generator-easy-direction}$]
		Let $B$ be the following promise problem:
		\begin{itemize}
			\item Input: A $(w, 1)$-automaton $Q$, states $q, r \in [w]$, a positive integer $t < 2^{\lceil \log^{1 + \eta} w \rceil}$, and padding to make the input length $2^{\lceil \log^{1 + \eta} w \rceil}$.
			\item Yes instances: $\Pr[Q^w(q; U_w) = r] \geq (t + 1)/2^{\lceil \log^{1 + \eta} w \rceil}$
			\item No instances: $\Pr[Q^w(q; U_w) = r] \leq (t - 1)/2^{\lceil \log^{1 + \eta} w \rceil}$.
		\end{itemize}
		Then $B \in \mathbf{promise\mhyphen BPL}$. Proof: Simulate $w$ steps of $Q$ from start state $q$ a total of $v$ times, using fresh randomness each time, where
		\[
			v \geq \frac{\ln 6}{2} 2^{2 \lceil \log^{1 + \eta} w \rceil}.
		\]
		Count how many end up in state $r$, and accept if and only if the fraction is at least $t/2^{\lceil \log^{1 + \eta} w \rceil}$. The space required by this algorithm is $O(\log^{1 + \eta} w)$, which is logarithmic in terms of the input length. By Hoeffding's inequality, this algorithm succeeds with probability at least $\frac{2}{3}$.
		
		Therefore, by the premise of the lemma, $B \in \mathbf{promise\mhyphen DSPACE}(\log^{1 + \gamma/(1+ \eta)} n)$, i.e. $B$ can be decided in deterministic space 
		\[
			O(\log^{(1 + \eta)(1 + \gamma/(1 + \eta))} w) = O(\log^{1 + \eta + \gamma} w).
		\]
		Let $\mathsf{Gen}_w: \{0, 1\}^{\lceil \log^{1 + \eta} w \rceil} \to \{0, 1\}^{2^{\lceil \log^{1 + \eta + \gamma} w} \rceil}$ be the identity function padded with zeroes. When the algorithm $\mathsf{S}$ is given $(Q, q, x)$ with $x \in \{0, 1\}^{\lceil \log^{1 + \eta} w \rceil}$ (i.e. discarding the padding), it behaves as follows:
		\begin{enumerate}
			\item Interpret $x$ as an integer in $\{0, \dots, 2^{\lceil \log^{1 + \eta} w \rceil} - 1\}$. Initialize $t = 0$.
			\item For each $r \in [w]$:
			\begin{enumerate}
				\item Find the largest $\Delta t$ such that $(Q, q, r, \Delta t)$ is accepted by the deterministic algorithm that decides $B$ (when the input is padded appropriately.)
				\item Set $t := t + \Delta t$.
				\item If $t \geq x$, output $r$.
			\end{enumerate}
		\end{enumerate}
		The space usage of $\mathsf{S}$ is $O(\log^{1 + \eta + \gamma} w)$ as it should be. Now we analyze the error. The probability of outputting a particular $r \in [w]$ when $x$ is chosen uniformly at random is precisely $\Delta t / 2^{\lceil \log^{1 + \eta} w \rceil}$, where $\Delta t$ is the largest value such that $(Q, q, r, \Delta t)$ is accepted by the algorithm that decides $B$ (when padded appropriately.) By the definition of $B$, this probability is within $2^{-\log^{1 + \eta} w}$ of $\Pr[Q^w(q; U_w) = r]$. Total variation distance is half $L_1$ distance, so the error of the simulator is at most $\epsilon = \frac{1}{2} w 2^{-\log^{1 + \eta} w}$. Hence $\log(1/\epsilon) \geq \Omega(\log^{1 + \eta} w)$.
	\end{proof}
	
	\section{Transforming targeted PRGs into advice generators in the uniform setting} \label{sec:uniform}
	
	The proof of our main result (Theorem~\ref{thm:targeted-prg-to-advice-generator}) is complete; this section can be considered ``optional reading''. In this section, we give an unconditional proof of a uniform statement analogous to Condition~\ref{cond:targeted-prg-to-advice-generator} in Theorem~\ref{thm:targeted-prg-to-advice-generator}. Namely, we show that targeted pseudorandom generators can be transformed into simulation advice generators, as long as we only worry about correctness with respect to sequences of automata that can be generated in logspace.
	
	One might hope that this would lead to an unconditional derandomization of $\mathbf{BPL}$ that is only guaranteed to work for easily-generated inputs. Unfortunately, we are not able to prove such a result: when trying to simulate an easily-generated automaton $Q$ using the SZA transformation, the approximate powers of $Q$ that arise are not so easily generated.
	
	\begin{defn}
		Suppose $((Q_1, q_1), (Q_2, q_2), \dots)$ is a sequence where $Q_w$ is a $(w, 1)$-automaton and $q_w \in [w]$. We say that the sequence is \emph{uniform} if there is some deterministic algorithm that, given $w$, produces $(Q_w, q_w)$ in space $O(\log w)$.
	\end{defn}
	
	\begin{defn}
		We say that $\mathsf{Gen}_w$ is\footnote{Again, strictly speaking, this is a property of a \emph{family} $\{\mathsf{Gen}_w\}$, not an individual generator.} a \emph{targeted $\epsilon$-pseudorandom generator against $\mathbf{Q}_{w, 1}^{m}$ in the uniform setting} if $\mathsf{Gen}_w$ is a targeted $\epsilon$-pseudorandom generator against $\mathbf{A}_w \subseteq \mathbf{Q}_{w, 1}^{m}$ such that for every uniform sequence $((Q_1, q_1), (Q_2, q_1), \dots)$, for all sufficiently large $w$, the element of $\mathbf{Q}_{w, 1}^{m}$ specified by $(Q_w, q_w)$ is an element of $\mathbf{A}_w$. We similarly define what it means for $\mathsf{Gen}_w$ to be an \emph{$\epsilon$-simulation advice generator for $\mathbf{Q}_{w, 1}^{m}$ in the uniform setting}.
	\end{defn}
	
	\begin{prop} \label{prop:uniform}
		For any constant $\mu \in [0, 1]$ and for any constants $\sigma > \eta > 0$, the following holds. Suppose there is a family $\{\mathsf{Gen}_w\}$, where $\mathsf{Gen}_w$ is an efficiently computable targeted $\epsilon$-pseudorandom generator against $\mathbf{Q}_{w, 1}^m$ in the uniform setting with seed length $s$ satisfying
		\begin{align*}
		&s \leq O(\log^{1 + \sigma} w),
		&&\log(1/\epsilon) = \log^{1 + \eta} w,
		&&&\log m \geq \log^{\mu} w.
		\end{align*}
		Then there is another family $\{\mathsf{Gen}'_w\}$, where $\mathsf{Gen}'_w$ is an efficiently computable $\epsilon$-simulation advice generator for $\mathbf{Q}_{w, 1}^m$ in the uniform setting with seed length $s$ and output length $a' \leq \poly(w)$.
	\end{prop}
	
	The proof of Proposition~\ref{prop:uniform} is simple: the simulation advice is just a list of pseudorandom strings for particular $(Q, q)$ pairs. The length of the list is small, but $\omega(1)$, and constructed in such a way that for any uniform sequence $((Q_1, q_1), (Q_2, q_2), \dots)$, for sufficiently large $w$, the advice includes a pseudorandom string for $(Q_w, q_w)$.
	
	\begin{proof}
		$\mathsf{Gen}'_w$ behaves as follows, given seed $x$:
		\begin{enumerate}
			\item For all programs $P$ of length at most $\log w$ that on input $w$ have an explicit self-imposed $O(\log w)$ space bound:
			\begin{enumerate}
				\item Run $P(w)$. If it produces a pair $(Q, q)$ where $Q$ is a $(w, 1)$-automaton and $q \in [w]$, then print $(Q, q, \mathsf{Gen}(Q, q, x))$.
			\end{enumerate}
		\end{enumerate}
		This generator clearly uses space $O(\log^{1 + \sigma} w)$, has seed length $s$, and has output length $a \leq \poly(w)$. The corresponding algorthm $\mathsf{S}$ behaves as follows, given $(Q, q, y)$ where $y$ is the output of $\mathsf{Gen}'_w$:
		\begin{enumerate}
			\item If, for some $z$, the triple $(Q, q, z)$ appears in $y$, then output $Q^{|z|}(q; z)$. Otherwise output $1$.
		\end{enumerate}
		This algorithm clearly runs in logspace. We will show that it is an $\epsilon$-simulator for $\mathbf{Q}_{w, 1}^m$ in the uniform setting. Indeed, suppose $((Q_1, q_1), (Q_2, q_2), \dots)$ is uniform via some program $P$. Then for all sufficiently large $w$, $\mathsf{Gen}_w$ works against $(Q_w, q_w)$. Furthermore, when $w \geq 2^{|P|}$, the algorithm for $\mathsf{Gen}'_w$ will consider $P$, and hence its output will include the triple $(Q_w, q_w, \mathsf{Gen}_w(Q_w, q_w, x))$. Therefore, for such $w$, the simulator will give an output that is $\epsilon$-close to $Q_w^m(q_w; U_m)$.
	\end{proof}
	
	\section{Acknowledgments}
	The first author is supported by the National Science Foundation Graduate Research Fellowship under Grant No. DGE-1610403. The second author is supported by National Science Foundation Grant No. CCF-1423544 and by a Simons Investigator grant.
	
	\bibliographystyle{alpha}
	\bibliography{iterated-sza}
\end{document}